\documentclass[12pt]{article}
\usepackage{ae} 
\usepackage[T1]{fontenc}
\usepackage{amssymb}
\usepackage{latexsym}
\usepackage{amsmath}
\usepackage[ansinew]{inputenc}
\usepackage{euscript}
\usepackage{epsfig}
\usepackage{epsf}
\usepackage{psfrag}
\usepackage{multirow}
\usepackage{setspace}
\usepackage{anysize}
\usepackage{amsmath}

\newcommand{\bc}{\begin{center}}
\newcommand{\ec}{\end{center}}

\newcommand{\qed}{\hspace*{\fill}
            $\Box$\smallskip}
\newcommand{\thmqed}{\hspace*{\fill}
            $\blacksquare$\smallskip}

\newenvironment{proof}{\noindent {\bf Proof:} \par}
                      {\qed}
\newenvironment{thmproof}{\noindent {\bf Proof:} \par}
                      {\thmqed}

 \newtheorem{theorem}{Theorem}
 
 \newtheorem{lemma}{Lemma}[section]
 \newtheorem{example}{Example}
 
 \newtheorem{definition}{Definition}
 \newtheorem{proposition}{Proposition}
\begin{document}

\begin{titlepage}

\bc
{\Large\bf Equilibrium and Potential in Coalitional Congestion Games}%
\footnote{This work is based on Sergey Kuniavsky's M.Sc thesis done
under the supervision of Rann Smorodinsky. Financial support by the
Technion's fund for the promotion of research and the Gordon Center for System Engineering is gratefully acknowledged. Valuable comments by an anonymous referee are gratefully acknowledged.}
\\[0.5cm] by
\\[0.5cm] {\large\bf
Sergey Kuniavsky%
\footnote{Munich Graduate School of Economics, Kaulbachstr. 45, Munich 80539, Germany. Financial support from the Deutsche Forschungsgemeinschaft through GRK 801 is gratefully acknowledged. $<${\tt
Sergey.Kuniavsky@lrz.uni-muenchen.de.}$>$.} and Rann Smorodinsky%
\footnote{Corresponding
author: Faculty of Industrial Engineering and Management, Technion,
Haifa 32000, Israel. $<${\tt rann@ie.technion.ac.il $>$}}}
\\[0.5cm]

\ec

\vspace*{1\baselineskip}

\bc {\bf Abstract} \ec The model of congestion games is widely used
to analyze games related to traffic and communication. A central
property of these games is that they are potential games and hence
posses a pure Nash equilibrium. In reality it is often the case that
some players cooperatively decide on their joint action in order to maximize the coalition's total utility.
This is by modeled by Coalitional Congestion Games. Typical settings include truck drivers who work for
the same shipping company, or routers that belong to the same ISP. The formation of coalitions will typically imply that the resulting coalitional congestion game will no longer posses a pure Nash equilibrium. In this paper we provide conditions under which such games are potential games and posses a pure Nash equilibrium.

\vspace*{2cm}


\end{titlepage}
\renewcommand{\thefootnote}{\arabic{footnote}}
\setcounter{footnote}{0}

\section{Introduction}

Congestion games, introduced by Rosenthal \cite{Rosenthal}, form a very natural model for studying many real-life strategic settings: Traffic problems, load balancing, routing, network planning, facility locations and more. In a congestion game players must choose some subset of resources from a given set of resources (e.g., a subset of edges leading from the Source to the Target on a graph). The congestion of a resource is a function of the number of players choosing it and each player seeks to minimize his total congestion accross all chosen resources.
In many modeling instances players and the decision making entity have been thought of as one and the same. However, in a variety of settings this may not be the case.

Consider, for example, a traffic routing game where each driver chooses his route in order
to minimize his travel time, while accounting for congestion along the route caused by other drivers.
Now, in many cases drivers are actually employees
in shipping firms, and in fact it is in the interest of the shipping firm to minimize the total travel time of its fleet.
Similarly, routers in a communication network participate in a congestion game. However, as various routers may belong to the same ISP we are again in a setting where coalitions naturally form. This motivated Fotakis et al. \cite{Fotakis} and Hayrapetyan et al. \cite{Hayra} to introduce the notion of Coalitional Congestion Games (CCG). In a CCG we think about the coalitions as players and each coalition maximizes its total utility.
The coalitional congestion game inherits its structure from the original game,
once the coalitions of players from the original congestion games
(now, becoming the players of the coalitional congestion game) have been identified.

The most notable property of congestion games is that they have posses pure Nash equilibrium. This has been shown by Rosenthal in \cite{Rosenthal}. Later, Monderer and Shapley \cite{Monderer Shapley 96}
formally introduce potential games and show the equivalence of these two classes. The fact that potential games posses a pure Nash equilibrium is straightforward. Unfortunately, the statement that a CCG is a potential game or that it possesses a pure Nash equilibrium is generally false.
In this paper we investigate conditions under which this statement is true. We focus on a subset of congestion games called simple congestion games, where each player is restricted to choose a single resource.

Our main contributions are:

\begin{enumerate}
\item
    Whenever each coalition contains at most two players the CCG induced from a simple congestion game possesses a pure-strategy Nash equilibrium (Theorem \ref{Pairs Pure NE}).
\item
If some coalition contains three players, then there may not exist a pure-strategy
Nash equilibrium (Example \ref{No NE}).
\item
If a the congestion game is not simple then there may not exist
a pure-strategy Nash equilibrium (Example \ref{not SCG underlying});
and
\item
Suppose there exists at least one singleton coalition and at least one coalition composed of two players, then a coalitional congestion game induced from
a simple congestion game is a potential game if and only if cost functions are linear (Theorem \ref{linear costs}).
\end{enumerate}

Our results extend and complement the results in Fotakis et al. \cite{Fotakis} and Hayrapetyan et al. \cite{Hayra}.
For example, Fotakis et al. \cite{Fotakis} show that if the resource cost
functions are linear then the coalitional congestion game is a potential
game. We show that, with some additional mild conditions on the
partition structure, this is also a necessary condition.
Hayrapetyan \cite{Hayra} shows that if the underlying congestion
game is simple and costs are weakly convex then the game possesses a pure
Nash equilibrium. We demonstrate additional settings where this holds.

Section 2 provides a model of a coalitional congestions games, section 3 discusses the conditions for the existence of a pure Nash equilibrium in such games and section 4 discusses the conditions for the existence of a potential function.
All proofs are relegated to an appendix.

\section{Model}

Let $G=\{N,S,U\}$ be a non-cooperative game in strategic form. Let
$C=\{C_1, \ldots, C_{n^c}\}$ be a Partition of $N$ into $n^c$
nonempty sets. Hence: $\cup_{k=1}^{n^c} C_k = N$ and $C_k \cap C_l
= \emptyset \;\; \forall k \neq l \in [1, \ldots n^c]$.

The game $G$ and the partition $C$ form a Coalitional
Non-Cooperative (CNC) Game $G^C = \{N^C,S^c,U^c\}$ defined as
follows:\begin{itemize} \item $N^C$ is the set of agents which are
the elements of $C$. \item The strategy space is $S^c =
\{S^c_k\}_{k \in C}$ where $S^c_k = \times_{i \in C_k}S_i$.
\end{itemize}
Note that $\times S_k = \times_{k \in C}\times_{i \in C_k}S_{k,i}$
is isomorphic to $S=\times_{i=1}^n S_i$, since we only changed the
order of the coordinates. Thus, we can look on $s^c$ as a vector
in $S$.
\begin{itemize}
\item The utility function is defined as follows: $\forall s^c \in
S^c \;\; U^c_k(s^c) = \sum_{ i \in C_k}U_i(s^c)$ and $U^c
=\{U^c_k\}_{k \in C}$.
\end{itemize}

In the context of $G^C$, $G$ is the \emph{Underlying Game} and a
player in $G$ is referred to as a  \emph{sub agent}. As always, a
Pure Nash Equilibrium of the game $G^C$ is a strategy profile
$s \in S^C$ such that $\forall k \in N^C$, $U^c_k(s) \geq
U^c_k(s_{-k},t_k) \;\;\forall t_k \in S^c_k$. $NE(G^C)$ the
(possibly empty) set of Pure Nash equilibrium strategy profiles in $G^C$.



A congestion game is a game $G=\{N,R,\Sigma,P\}$ where $N$ is the
finite set agents, $R$ is the finite set resources, $P=\{P_r\}_{r
\in R}$ are the resource costs functions, where $P_r:[1, \ldots, n]
\to \mathbb{R}$ and $\Sigma=\times_{i \in N}\Sigma_i$, where
$\Sigma_i \subseteq 2^R$, is $i$'s strategy space. Agent $i$ selects
$s_i \in \Sigma_i$ and pays $\sum_{r \in s_i}P_r(c(s)_r)$, where
$c(s)_r = \sum_{j \in N}\mathbb{I}_{\{r \in s_j\}}$ is the number of
agents who select $r$ in $s$. In utility terms, the utility of agent
$i$ is $U_i(s) = -\sum_{r \in s_i}P_r(c(s)_r)$. If $\Sigma_i = R
\;\; \forall i \in N$ then $G=\{N,R,\Sigma,P\}$ is called a {\it
Simple Congestion Game}.

We will assume that the $P_r$ functions are non-negative and
increasing ($P_r(1)$ can be zero).

Fix a Simple Congestion Game or Coalitional Congestion Game (CCG) $G$ with $R$ resources and $n$ (sub-)agents. A
\emph{congestion vector} is an element of $\mathbb{N}^R$ whose
elements sum up to $n$. A Simple Congestion Game (CCG) $G$ and a strategy profile $s$
induce a congestion vector $c(s)$: $\{c(s)_r\}_{r \in R}$.

A strategy profile $s$ of a Coalitional Congestion Game $G^C$
induces a private congestion vector for each of the agents in $G^C$.
Such vector for agent $k$ will be $c_k$: $c_k(s_k)_r = |\{i \in C_k
: s_{k,i}=r\}|$ which is an element if $\mathbb{N}^R$ whose elements
sum up to $|C_k|$.

Let $X$ be a subset of the strategy profiles space. We denote
$c(X)$ as the corresponding set of congestion vectors: $\forall X
\subseteq S \;\; c(X) = \{c(s) \;s.t: s \in X \}$

\section{CCG and Pure Nash Equilibria}

The following preliminary result asserts that if the a Nash equilibrium the underlying game is composed of strategies such that all sub-agents of
any agent choose different resources, then it is also an equilibrium of the coalitional game.

\begin{proposition}\label{layer NE}
Let $G$ be a Simple Congestion Game , $C$ a partition of $N$ and $G^C$ be the induced CCG.
Let $s$ be a
strategy profile of $G^C$ where $s_{k,i} \neq s_{k,j} \;\; \forall
k \in N^C, \forall i,j \in C_k$. If $c(s) \in c(NE(G))$ then $s \in NE(G^C)$.
\end{proposition}

This is key to proving our central result about the existence of a Nash equilibrium in CCGs:

\begin{theorem}\label{Pairs Pure NE}
Let $G$ be a Simple Congestion Game , $C$ a partition where the largest element is of size 2, and $G^C$ a
CCG with the underlying game $G$ and the partition $C$. Then
$NE(G^C) \neq \emptyset$. That is, if the largest coalition is a Pair,
a Pure Nash equilibrium always exists.
\end{theorem}

Does this existence result extend to other partition forms, where the maximal element has
more than two sub-agents? The following example demonstrates that this is not true in general:

\begin{example}\label{No NE}
Consider a game with two identical resources A and B and four sub
agents, with the following payment functions:

\begin{center}
\begin{tabular}{|c|c|c|c|c|}
\hline
Resource / Agents \#: & 1 & 2 & 3 & 4 \\
\hline
A: & 0 & 12 & 16 & 18 \\
\hline
B: & 0 & 12 & 16 & 18 \\
\hline
\end{tabular}
\end{center}

When $C=[\{1,2,3\},\{4\}]$ the matrix form of the resulting 2-player CCG is:
\begin{center}
\begin{tabular}{|c|c|c|}
\hline
$G^{C}$ & A & B\\
\hline
A,A,A & -54, -18 & -48, 0 \\
\hline
A,A,B & -32, -16 & -36, -12  \\
\hline
A,B,B & -36, -12 & -32, -16  \\
\hline
B,B,B & -48, 0 & -54, -18  \\
\hline
\end{tabular}
\end{center}

Whereas the underlying Simple Congestion Game  has a pure Nash equilibrium this CCG has none. To verify this note for the compound agent (made up of 3 sub-agents) the strategies $AAA$ and $BBB$ are dominated. Following their deletion the remaining game is one of matching pennies and hence has no
Pure Nash equilibrium .
\end{example}

Can the result of Theorem \ref{Pairs Pure NE} be extended to CCGs
with small coalition size, but with an underlying congestion game
that is not simple? Again, the answer is negative:

\begin{example}\label{not SCG underlying}
Let $G$ be a Congestion Game with three identical resources and
three agents. Each agent of $G$ chooses two of the three
resources. The cost of each resource is $P(n)=6-\frac{6}{n}$.

Let $C = [\{1,2\}\{3\}]$, and $G^C$ is the CCG with underlying
game $G$ and partition $C$. After omitting identical strategies
due to sub agents symmetry $G^C$ looks as follows:

\begin{center}
\begin{tabular}{|c|c|c|c|}
\hline
$G^C$ & AB & AC & BC\\
\hline
AB,AB & -16,-8  & -14,8 & -14,-4  \\
\hline
AC,AC & -14,-4 & -16,4 & -14,-4  \\
\hline
BC,BC & -14,-4 & -14,-4 & -16,-8  \\
\hline
AB,AC & -11,-7 & -11,-7 & -12,-6  \\
\hline
AB,BC & -11,-7 & -12,-6 & -11,-7  \\
\hline
AC,BC & -12,-6 & -11,-7 & -11,-7  \\
\hline
\end{tabular}
\end{center}
Note that compound agent's strategies (AB,AB), (AC,AC) and
(BC,BC) are dominated. Note that the remaining game has no pure Nash equilibrium.
\end{example}

\section{CCG and Potential}

An {\em Exact Potential} is a function $\mathbb{P}:S \to \mathbb{R}$ satisfying:
\begin{eqnarray}\label{Potential Definition}
\mathbb{P}(s) - \mathbb{P}(s_{-i}, t_i) =
U_i(s) - U_i(s_{-i},
t_i) \\ \nonumber \forall i \in N, \forall t_i \in S_i, \forall s \in S_1
\times S_2 \ldots \times S_n
\end{eqnarray}

Games with a potential function are called {\em Potential Games}. It is well known that potential games have a pure Nash equilibrium (see Monderer and Shapley \cite{Monderer Shapley 96}). In particular Congestion Games are potential games. Fotakis et al \cite{Fotakis} prove that a CCG,
where the cost functions of the resources of the underlying game are linear,  is a potential game.

Removing the linearity assumption is problematic. In fact, even in
the case of a CCG with a maximal coalition of size $2$,
which guarantees the existence of a pure Nash equilibrium (Theorem
\ref{Pairs Pure NE}), the existence of Exact Potential is not
guaranteed. In the following example we show that the existence of a potential function implies linearity of the cost functions:

\begin{example}\label{no potential}
Consider a congestion games with 2 resources, $A$ and $B$, with costs $a_1, a_2, a_3$ and $b_1, b_2, b_3$, correspondingly. Assume there are 3 players and set the coalition structure to  $C=[\{1,2\}\{3\}]$.
This induces the following two player CCG, given in matrix form:

\begin{center}
\begin{tabular}{|c|c|c|}
\hline
$G^{C}$ & A & B\\
\hline
A,A & $2a_3, a_3$ & $2a_2, b_1$ \\
\hline
A,B & $a_2+b_1, a_2$ & $a_1+b_2, b_2$\\
\hline
B,B &  $2b_2, a_1$ & $2b_3, b_3$\\
\hline
\end{tabular}
\end{center}

Assume this game has an exact potential with the following values:
\begin{center}
\begin{tabular}{|c|c|c|}
\hline
$G^{C}$ & A & B\\
\hline
A,A & $P_1$ & $P_2$ \\
\hline
A,B & $P_3$ & $P_4$\\
\hline
B,B &  $P_5$ & $P_6$\\
\hline
\end{tabular}
\end{center}

From the definition of exact potential the following must hold (see, in addition, Theorem 2.9. in Monderer and Shapley \cite{Monderer Shapley 96}):

\begin{eqnarray*}
(a_2+b_1-2a_3)+(a_3-b_1)+(2a_2-a_1-b_2)+(b_2-a_2)=\\ (P_3-P_1)+(P_1-P_2)+(P_2-P_4)+(P_4-P_3) = 0.\\
\end{eqnarray*}
Similarly:
\begin{eqnarray*}a_2+b_1-2b_2+a_1-b_3+2b_3-a_1-b_2+b_2-a_2=0\\
2a_3-2b_2+a_1-b_3+2b_3-2a_2+b_1-a_3=0.
\end{eqnarray*}

Manipulating these equalities leads to:
\begin{eqnarray*}
2a_2=a_1+a_3\\
2b_2=b_1+b_3,
\end{eqnarray*}
which implies that the cost functions are linear.

\end{example}

Using this example we can now prove our final result:

\begin{theorem}\label{linear costs}
Let $G$ be a Simple Congestion Game and let $C$ be a partition that has at least one element of
size 1 and at least one element of size 2. Let $G^C$ be a CCG with
the underlying game $G$ and partition $C$. $G^C$ will posses an
Exact Potential iff the CCG is linear.
\end{theorem}

\begin{thmproof}
Sufficiency - This has been obtained Fotakis et al. \cite{Fotakis} (Theorem 6).

Necessity - Recall that example \ref{no potential} provides a 3 player congestion game and a coalition structure that yields a CCG for which linear cost function are necessary for the existence of a potential.
The reason that the the linearity extends beyond the example to all situations implied in the theorem is that for any general CCG we can fix the strategy for all but 2 agents, of which one has 2 sub agents and one has a single sub agent. We can now look at the induced 2 player game. If the original game was a potential game so must be the induced game. The example then implies linearity in the induced game. However, as we can arbitrarily fix the strategy for all but the relevant 3 sub agents the result follows.
\end{thmproof}

\begin{appendix}
\section{Appendix - Omitted Proofs}

We begin with the definition of an auxiliary game. Let $G$ be a Simple Congestion Game  and let $C$ be a partition. The {\em Restricted Coalitional Congestion Game}, denoted $\overline{G^C}$, is a CCG where coalitions are restricted strategies such that distinct sub-agents choose distinct resources. More formally:

\begin{definition}
Let $G$ be a Simple Congestion Game  and let $C$ be a partition. The {\em Restricted Coalitional Congestion Game}, denoted $\overline{G^C}$, is the game $\overline{G^C}= \{N^c,\overline{S^c},U^c\}$, where $N^c$ and $U^c$
are as before and $\overline{S^c} = \{\overline{S^c_k}\}_{k \in C}$
where $\overline{S^c_k} = \{\times_{i \in C_k}S_{k,i} : s_{k,i}\neq
s_{k,j} \;\; \forall i,j \in C_k$\}.
\end{definition}

The following result about pure Nash equilibria in restricted coalitional congestion games will be useful for proving our main result:

\begin{lemma}\label{restricted layer NE}
Let $G$ be a Simple Congestion Game  and $C$ a partition of $N$. Let $\overline{G^C}$ be
a Restricted CCG with the underlying game $G$ and the partition $C$.
Let $s$ be a strategy profile of $\overline{G^C}$ (where $s_{k,i}
\neq s_{k,j} \;\; \forall k \in N^C, \forall i,j \in C_k$). If $c(s)
\in c(NE(G)) \Rightarrow s \in NE(\overline{G^C})$
\end{lemma}

\begin{proof}
For two congestion vectors $u,v$, let $d(u,v) = \frac{\sum_{r\in R} |{v_r-u_r}|}{2}$, denote the distance between these two vectors.

Let $s$ be a strategy profile satisfying the condition of the lemma,
let $k$ be an arbitrary agent in the game $\overline{G^C}$ and
denote by $s_{-k}$ be the strategy profile of all players except
$k$. We denote by $BR(s_{-k})$ the set of $k$'s best reply
strategies to $s_{-k}$. Assume , by way of contradiction, that $s_k
\not \in BR(s_{-k})$ and let $t_k \in BR(s_{-k})$ be a best reply to
$s_{-k}$ which corresponding congestion vector has a minimal
distance to $c(s_k)$. Namely, $d(c(t_k),c(s_k)) \le
d(c(t'_k),c(s_k)) \ \ \forall t'_k \in BR(s_{-k})$.

In $\overline{G^C}$ each agent selects each resource at most once.
As $c_k(t_k) \neq c_k(s_k)$ this implies that there are two resources, say
$r$ and $x$, such that: $(c_k(t_k)_r, c_k(t_k)_x)=(1,0)$ and $(c_k(s_k)_r, c_k(s_k)_x)=(0,1)$ and, in addition, that $
c(s)_r = c(t_k,s_{-k})_r - 1$ and $c(s)_x = c(t_k,s_{-k})_x+1
$.

Let $i$ be the sub-agent of $k$ that chooses $r$ in $t_k$ but chooses a resource different than $r$ in $s_k$. Let $t'_k$ be a strategy for agent $k$, derived from $t_k$ by moving sub
agent $i$ from $r$ to $x$. This results in $c(t'_k,s_{-k})_x=c(s)_x$  and $c(t'_k,s_{-k})_r=c(s)_r$.
By assumption, $c(s) \in NE(G)$. Therefore $P_x(c(s)_x) \leq P_r(c(s)_r+1)$ and so:
\begin{equation}
P_x(c(t'_k,s_{-k})_x)=P_x(c(s)_x) \leq P_r(c(s)_r+1) =
P_r(c(t_k,s_{-k})_r)
\end{equation}

Thus, the contribution of sub agent $i$ to agent $k$'s payment in $(t'_k,s_{-k})$ is less or equal its contribution to $k$'s payment in $(t_k,s_{-k})$. As the only change in $k$'s strategy between $t_k$ and $t'_k$ is $i$'s choice, we conclude that $k$'s payment in $(t'_k,s_{-k})$ is less or equal its payment in $(t_k,s_{-k})$, and so $t'_k \in BR(s_{-k})$. However, by construction $d(c(t'_k),c(s_k)) \le d(c(t_k),c(s_k))$, contradicting the way $t_k$ was chosen.
\end{proof}.

{\bf Proof of Proposition \ref{layer NE}}:

Let $s$ be a profile as described in the Proposition. Let $t_k$ be the
best reply strategy for agent $k$ to $s_{-k}$. We show that
$c_k(t_k)_r \leq 1 \;\forall k \in N^C$ and $\forall r \in R$.

Assume this is not true and that some resource $r$,
$t_{k,i}=t_{k,j}=r$. Since $c_k(s_k)_r \leq 1$ we know that
$c(s_{-k},t_k)_r > c(s)_r$. Therefore there must exist some resource
$x$ such that $c(s_{-k},t_k)_x < c(s)_x$.

Let $t'_k$ be a strategy profile derived from $t_k$ by moving sub
agent $i$ from $r$ to $x$. In the strategy profile $(s_{-k},t'_k)$
agent $i$ pays $P_x(c(s_{-k},t'_k)_x)$. By construction:
\begin{equation}\label{31eq1}
c(s_{-k},t'_k)_x=c(s_{-k},t_k)_x+1 \leq c(s)_x
\end{equation}
From  Equation \ref{31eq1} and the fact that $c(s) \in NE(G)$ we get
that:
\begin{equation}\label{31eq2}
P_x(c(s_{-k},t'_k)_x) = P_x(c(s_{-k},t_k)_x+1) \leq P_x(c(s_x)_x)
\leq P_r(c(s)_r+1)
\end{equation}
Using monotonicity of the cost functions and the fact that $c(s)_r+1
\leq c(s_{-k},t_k)_r$ we get:
\begin{equation}\label{31eq3}
P_r(c(s)_r+1) \leq P_r(c(s_{-k},t_k)_r)
\end{equation}
Combining Equations \ref{31eq2} and \ref{31eq3} we get that
\begin{equation}\label{31eq4}
P_x(c(s_{-k},t'_k)_x) \leq P_r(c(s_{-k},t_k)_r)
\end{equation}

We will now show that $k$ is better off in the strategy profile
$(s_{-k},t'_k)$ than in $(s_{-k},t_k)$, thus contradicting the fact
that $t_k$ is a best response to $s_{-k}$. We do this but analyzing
each of $k$'s sub-agents:
\begin{itemize}
\item
Sub agent $i$ pays in $(s_{-k},t'_k)$, where he chose $x$, no more
than than in $(s_{-k},t_k)$, where he chose $r$ (equation
\ref{31eq4}).
\item
From definition of $x$, $c_k(t_k)_x<c_k(s_k)_x$. Since
$c_k(s_k)_x=1$ and $c_k(s_k)_x>c_k(t_k)_x$, we get that
$c_k(t_k)_x=0$. Thus, apart from agent $i$ no other sub agent of $k$
chose $x$ in $t_k$.
\item
Sub agent $j$, who selects $r$ both in $(s_{-k},t'_k)$ and
$(s_{-k},t_k)$, pays strictly less in $(s_{-k},t'_k)$ than in
$(s_{-k},t_k)$, because $c(s_{-k},t'_k)_r<c(s_{-k},t_k)_r$. This
inequality holds for any other sub agent of $k$ who chose $r$ in
$t_k$
\item
All sub agents who choose a resource in the set $R \setminus\{r,x\}$
pay the same in $(s_{-k},t_k)$ and $(s_{-k},t'_k)$.
\end{itemize}
To conclude, agent $k$ pays strictly less in $(s_{-k},t'_k)$ than in
$(s_{-k},t_k)$. This contradicts the fact that $t_k$ is a best reply
to $s_{-k}$. Therefore, any best reply of $k$ to $s_{-k}$ must be
such that all the sub-agents choose different resources.

Thus, agent $k$'s best reply strategy to $s_{-k}$ is a strategy that
is allowed also in $\overline{G^C}$. We couple this observation with
the observation that $s$ is a Nash equilibrium of $\overline{G^C}$
(follows from Lemma \ref{restricted layer NE}) and the fact that $k$
is arbitrary to conclude that $s\in NE(G^C)$.

QED

{\bf Proof of Theorem \ref{Pairs Pure NE}:}

Let $s$ be an arbitrary Nash equilibrium of $G$

Case 1 - Assume that $c(s)_r \le |N^C|$ for all resources $r \in R$.
In this case we can re-arrange the players over the resources so
that the result will be a strategy profile with the same congestion
vector, and furthermore, for any $k \in N^C$ its two sub agents, $i,
j$, choose different resources. The resulting vector is also in
$NE(G)$ and complies with the conditions in Proposition \ref{layer NE}.
Therefore that proposition suggests that $s$ is a Nash equilibrium of
$G^C$.

Case 2 - Let us denote by $r$ the resource with the highest
congestion in $s$ and assume $c(s)_r>|N^C|$. We argue that without
loss of generality (by rearranging the players) $s$ has the
following two properties: (a) if agent $k$ has its 2 sub agents on
the same resource then it must be the case that the corresponding
resource is $r$, that is, $\forall k \in N^C$ $c(s_k)_{r'}>1$
implies $r'=r$; and (b) all agents have at least one sub-agent
choose $r$.

Note some properties of the strategy tuple $s$:
\begin{enumerate}
\item Let $k$ be an agent with a single sub agent. Then this sub agent must
be on $r$ and it has no profitable
deviation.

\item Let $k$ be an agent with a two sub agents, $i$ and $j$. Assume $i$
is on $r$ and $j$ is on some $r'\not = r$. Moving a single subagent cannot be
profitable.

\item Let $k$ be an agent with two sub-agents, $i$ and $j$. Assume $i$
is on $r$ and $j$ is on some $r'\not = r$. Moving both sub agents simultaneously cannot be profitable as at least one of these moves makes $k$ worse off, while the other cannot improve $k$
payoff.

\item Let $k$ be an agent with two sub-agents, $i$ and $j$, both on $r$.
Moving both sub agents cannot be
profitable.
\end{enumerate}

Thus, if $s$ is not a NE of $G^c$, the only profitable deviation possible is for an agent $k$ with two
sub-agents, $i$ and $j$, both on $r$, to move one sub-agent, say $j$ to another resource, say $r'$.
Furthermore, let us assume that this is the most profitable deviation for $k$. That is
$P_{r'}(c(s_{r'})+1) \leq P_{r''}(c(s_{r''})+1)$ for all $r'' \neq
r$. We denote the resulting strategy profile by $s'$. Note the
properties of $s'$:

\begin{enumerate}
\item All agents with a single sub-agent choose the resource $r$.
\item All agents with two sub agents, have at lease one sub agent in $r$.
\item The payment of all sub agents in $r$ is lower compared with $s$, while the payment of all subagents in $R\setminus \{r\}$ is at least as large compared with the payment in $s$.
\item Any agent with two sub agents, one on $r$ and one on some $r'' \not = r$ pays (weakly) less than what $k$ paid in $s$
\end{enumerate}

Assume $s'$ is not a Nash equilibrium, then there must be some profitable deviation. What are the possible profitable deviations?

\begin{enumerate}
\item Let $k'$ be an agent with a single sub agent. Then this sub agent must be on $r$ and it has no profitable deviation. Recall that the payment of $k$ is $s'$ is lower than the payment of $k'$ is $s$.

\item Let $k'$ be an agent with a two sub agents, $i$ and $j$. Assume $i$ is on $r$ and $j$ is on some $r'\not = r$. Moving $i$ cannot be profitable for the same argument as above. Moving $j$ to another resource in $R\setminus \{r\}$ cannot be profitable, so the only profitable deviation might be moving $j$ back to $r$. However, this would result $k'$ paying the same payment that $k$ paid in $s$, which by construction of $s'$ is higher than what $k$ pays in $s'$. Implying that $k'$ had a profitable deviation in $s$, thus contradicting what we already know.

\item Let $k$ be an agent with two sub-agents, $i$ and $j$. Assume $i$ is on $r$ and $j$ is on some $r'\not = r$. Moving both sub agents simultaneously cannot be profitable as at least one of these moves makes $k$ worse off, while the other cannot improve $k$ payoff.

\item Let $k$ be an agent with two sub-agents, $i$ and $j$, both on $r$. Moving both sub agents cannot be profitable.
\end{enumerate}

Once again, the only profitable deviation possible is for an agent $k'$ with
two sub-agents both on $r$, to move one sub-agent. The resulting
strategy profile $s''$, once more, only allows for profitable
deviations of the same form. Namely, for an agent $k'$ with two
sub-agents both on $r$, to move one sub-agent. We continue
iteratively in the same manner. As the process is bounded
by the number of agents selecting $r$ with both sub agents
in $s$, it must end in finitely many steps. The final strategy vector has no profitable deviations and is, therefore, a Nash equilibrium of the game.

QED

\end{appendix}


\begin{thebibliography}{label}

\bibitem{Fotakis05} Fotakis D., Kontogiannis S.,
Spiraklis P. ``Symmetry in Network Congestion Games: Pure
equilibria and Anarchy Cost," {\it Workshop on Approximation
and Online Algorithms (WAOA)}, pp. 161-175, 2006.

\bibitem{Fotakis} Fotakis D., Kontogiannis S.,
Spiraklis P. ``Atomic Congestion Games Among Coalitions,"
{\it  International Colloquium on Automata, Languages and Programming (ICALP)}, pp. 573-584, 2006.

\bibitem{Hayra} Hayrapetyan A., Tardos E., Wexler T.
``The Effect of Collusion in Congestion Games," {\it 38th annual
ACM symposium on Theory of computing (STOC)}, pp. 89-98, 2006

\bibitem{Monderer Shapley 96}
Monderer D. and Shapley L.S. ``Potential Games," {\it Games and
Economic Behavior}, 14, pp. 124-143, 1996.

\bibitem{Rosenthal} Rosenthal R.W. ``A class of games possesing Pure-Strategy Nash
Equilibria", {\it International Journal of Game Theory}, 2, pp. 65-67, 1973.

\end{thebibliography}
\end{document}